\documentclass[conference]{IEEEtran}
\usepackage{cite}
\usepackage{amsmath,graphicx}
\usepackage{amsthm}
\usepackage{xypic}
\usepackage{amssymb}
\usepackage{hyperref}
\usepackage[square]{natbib}
\setcitestyle{numbers}
\theoremstyle{plain}
\newtheorem{theorem}{Theorem}
\newtheorem{proposition}{Proposition}
\newtheorem{corollary}{Corollary}
\newtheorem{lemma}{Lemma}
\theoremstyle{definition}
\newtheorem{definition}{Definition}

\newtheorem{remark}{Remark}
\newtheorem{example}{Example}

\def\col{\mathcal}

\begin{document}

\title{Unsupervised clustering of file dialects according to monotonic decompositions of mixtures}
%

\author{
    \IEEEauthorblockN{Michael Robinson}
  \IEEEauthorblockA{Department of Mathematics and Statistics\\
    American University\\
    Washington, DC\\
    Email: michaelr@american.edu}\\
\and
      \IEEEauthorblockN{Tate Altman}
  \IEEEauthorblockA{Department of Mathematics and Statistics\\
    American University\\
    Washington, DC\\
    Email: ta8427a@american.edu}\\
\and
  \IEEEauthorblockN{Denley Lam}
  \IEEEauthorblockA{BAE Systems FAST Labs\\
    Arlington, VA\\
    Email: denley.lam@baesystems.com}\\
\and
  \IEEEauthorblockN{Letitia W. Li}
  \IEEEauthorblockA{BAE Systems FAST Labs\\
    Arlington, VA\\
    Email: letitia.li@baesystems.com}\\
}

\makeatletter
\def\ps@headings{%
\def\@oddhead{\mbox{}\scriptsize\rightmark \hfil \thepage}%
\def\@evenhead{\scriptsize\thepage \hfil \leftmark\mbox{}}%
\def\@oddfoot{\scriptsize \@date\hfil Approved for public release; distribution unlimited. Not export controlled per ES-FL-012723-0011.}%
\def\@evenfoot{\scriptsize Approved for public release; distribution unlimited. Not export controlled per ES-FL-012723-0011.\hfil \@date}}

\def\ps@IEEEtitlepagestyle{%
\def\@oddhead{\mbox{}\scriptsize\rightmark \hfil \thepage}%
\def\@evenhead{\scriptsize\thepage \hfil \leftmark\mbox{}}%
\def\@oddfoot{Approved for public release; distribution unlimited. Not export controlled per ES-FL-012723-0011.\hfil}%
\def\@evenfoot{Approved for public release; distribution unlimited. Not export controlled per ES-FL-012723-0011.\hfil}}

\makeatother

\pagestyle{headings}

\maketitle

\begin{abstract}
  This paper proposes an unsupervised classification method that partitions a set of files into non-overlapping dialects based upon their behaviors,
  determined by messages produced by a collection of programs that consume them.
  The pattern of messages can be used as the signature of a particular kind of behavior,
  with the understanding that some messages are likely to co-occur, while others are not.

  Patterns of messages can be used to classify files into dialects.
  A dialect is defined by a subset of messages, called the required messages.
  Once files are conditioned upon dialect and its required messages, the remaining messages are statistically independent.

  With this definition of dialect in hand,
  we present a greedy algorithm that deduces candidate dialects from a dataset consisting of a matrix of file-message data,
  demonstrate its performance on several file formats,
  and prove conditions under which it is optimal.
  We show that an analyst needs to consider fewer dialects than distinct message patterns,
  which reduces their cognitive load when studying a complex format.
\end{abstract}

\IEEEpeerreviewmaketitle

\section{Introduction}

While considerable effort has been expended to formalize what compliance with a format specification means,
the behavior of programs when files are consumed is what defines the end-user experience of a given format.
A behavioral understanding of file formats has the advantage that it is amenable to a statistical perspective,
wherein one can ascribe the conditional probability that a particular file will elicit a particular behavior given that other behaviors have already been observed.

The behavioral perspective aligns neatly with the discipline of test-driven design,
since files that elicit unwanted behaviors can easily be identified as test cases.
As such, curation of format-compliant file datasets is an important task for a file format analyst.
Files that are supposed to comply with a given \emph{ad hoc} format specification may in fact fall into one of several \emph{dialects},
in which different patterns of behavior can be observed.
Managing the behavioral differences between dialects is a source of trouble when one is attempting to construct programs to consume files of a given format.
To obtain adequate test coverage, one must ensure that all dialects are present in the test samples,
which means that an analyst must first know which files in their dataset comply with which dialects.
By partitioning the file format into dialects,
parser developers following the LangSec approach can develop grammars covering each dialect to develop more comprehensive parsers,
or they can formally define which dialects of a format their grammar should cover.

\subsection{Contributions}

This paper proposes an unsupervised classification method that partitions a set of files into non-overlapping dialects based upon their behaviors,
which are measured by the occurrence of a collection of Boolean features, called \emph{messages}.
The pattern of messages can be used as the signature of a particular kind of behavior,
with the understanding that some messages are likely to co-occur, while others are not.

Our method is based upon a novel statistical definition for a behavioral dialect.
A dialect is defined by a subset of messages, called the \emph{required messages}, that satisfies a statistical assumption.
Once files are conditioned upon a dialect and its required messages, the remaining messages are statistically independent.
The implications of this definition are detailed in Section \ref{sec:statistical_model}.

Our definition of dialect leads to a greedy algorithm that deduces candidate dialects from a dataset consisting of a matrix of file-message data.
This algorithm is embodied by the constructive proof of Proposition \ref{prop:decomp} in Section \ref{sec:methodology}.

The method we propose is able to work with files of any format, provided enough messages are available.
To highlight this fact, in Section \ref{sec:experimental} we explore our method's performance on three vastly different formats: a tabular data (CSV), free-form documents (PDF), and images (NITF).

Finally, we establish that our method is optimal in the sense that it yields the least number of extraneous dialects once certain reasonable statistical assumptions are satisfied.
These proofs of optimality appear in Theorems \ref{thm:min_decomp} and \ref{thm:minimal_count} in the Appendix, Section \ref{sec:appendix}.

\subsection{Limitations}

Our methodology relies upon good message coverage for the format under consideration.
If messages are not elicited by the behaviors of interest, then files which exhibit these behaviors cannot be detected.
Fortunately, most available parsers for the \emph{ad hoc} formats we consider in Section \ref{sec:experimental} produce copious output to {\tt stderr} and {\tt stdout}.
This output is sufficiently standardized that regular expressions (regexes) can be used to collect the output into messages.

If the set of messages being used is known to a malicious actor,
our approach might be subverted by crafting files to avoid producing certain messages.
Message regexes can be constructed in a semi-automated way,
outlined in detail by \cite{Robinson_langsec2022},
so it is not difficult to obtain enough messages of sufficient diversity to prevent files from evading proper classification.
Messages can also be made that correspond to system calls, resulting in additional behavioral diversity \cite{Scofield_2017}.
Anecdotally, it seems difficult to construct files that avoid \emph{all} messages of a certain type.

Our statistical model is a special case of a non-parametric independent mixture model.
Independent mixture models are very well studied in the literature,
with many algorithms that have deep theoretical backing (for instance \cite{Marin_2011, McLachlan_2000} among many others).
Unfortunately, these algorithms tend to make assumptions that are inappropriate for the context of file-message data,
such as assuming the components are of a known distribution or that the number of mixture components (dialects) is known from the outset.
Expectation maximization is a common tactic to avoid making these assumptions, at least when the number of features (messages) is small.

Our problem involves hundreds of features (messages),
which makes expectation maximization computationally infeasible \cite[1.3.1]{Marin_2011}.
To manage the complexity, we relax independence into a monotonicity condition.
As a result, instead of obtaining the best mixture decomposition, we merely obtain bounds upon it, see Theorem \ref{thm:minimal_count}.
Because our implementation uses a greedy algorithm,
there is some ambiguity that results in our candidate dialects,
though this is fairly benign.
When several messages are statistically dependent they need not all be chosen as required for a dialect, see Lemma \ref{lem:minimal_count} for details.

\subsection{Related work}

This paper continues a line of work presented over the past few years in the LangSec community that takes a statistical look at format specifications \cite{Robinson_langsec2022,Robinson_langsec2021,Ambrose_2020}.
In contrast to the hypothesis in \cite{Robinson_langsec2022}, that \emph{features are independent when conditioned upon dialect},
the present paper additionally conditions upon a set of messages that are required for each dialect.
When messages from different parsers are combined,
some of these messages are effectively identical features.
For instance, two parsers may emit the same kind of syntax error, resulting in two separate but statistically dependent messages.
This kind of behavior was ignored \cite{Robinson_langsec2022} even though it is a direct violation of the assumption of independence.
By conditioning upon one, the other, or both of these messages, we can restore independence.

Statistical format analysis appears to be a minority viewpoint,
because most file format analysis uses the structure of file \emph{contents} rather than the responses of parsers to those contents
(for instance, see \cite{belaoued2015real,al2018ransomware, 8685181,ALAZAB201591,demme2013feasibility}).
Nevertheless, statistical features based upon file actions has also been used to identify certain malicious behaviors \cite{Scofield_2017}.

We take inspiration from \cite{van2019wrangling}, in which $39$ dialects of CSV files were found.
Our methodology was applied to a simple random sample of the same dataset in Section \ref{sec:experimental},
wherein we find a somewhat coarser collection of $14$ dialects being most common, though numerous less common ones are also present.

From a mathematical perspective,
the methodology we use is based upon \emph{partially ordered sets} and the \emph{Dowker complex} \cite{bjorner1995topological}.
Recent work has connected the Dowker complex to tabular data \cite{Robinson_relations} and to formal concept analysis \cite{brun2023rectangle}.

\section{Statistical model of file format dialects and their behaviors}
\label{sec:statistical_model}

Our data consist of a set of files $F$,
which when parsed by a variety of programs may yield any of a certain set of messages $K$.
For each file, each message $k \in K$ either does occur (in which case we say that $k=1$), or does not occur (expressed as $k=0$).
Although parsing a given file is (usually!) deterministic,
we can model the likelihood of a given message occurring as a probability $P(k=1)$.
A probabilistic model avoids handling specific files individually,
so we rarely need to handle the set $F$ of files directly.

In \cite{Robinson_langsec2022,Robinson_langsec2021}, it was shown that looking at the joint probability of a set of messages either
occurring or not,
$P(k_1=1,k_2=1,\dotsc, k_n=1, k_{n+1}= 0, \dotsc)$,
was a useful way to identify certain files of interest.
In this probabilistic setting,
an event consists of a \emph{message pattern},
which is a subset of messages that might occur.

Studying arbitrary joint probabilities on their own is fraught,
though the data often support useful statistical assumptions that provide theoretical traction.
In \cite{Robinson_langsec2022}, it was assumed that messages for files within a given subset $A \subseteq F$
were independent.
While this is a reasonable assumption when messages are semantically unrelated,
it is not appropriate when some of the messages are related to each other.
In this article, we take a more refined approach.

\begin{definition}
  A \emph{dialect} is a subset of files $A \subseteq F$ and a subset of messages $R_A \subseteq K$,
  such that once conditioned on both $R_A$ and $A$,
  the remaining messages are independent.
  The subset $R_A$ is called the set of \emph{required messages} for dialect $A$.

  Explicitly, if we write $R_A = \{K_{i+1}, \dotsc, K_n\}$,
  the joint probability of a message pattern $K_1 = k_1$, $\dotsc$, $K_n = k_n$ 
  on a file in $A$ is of the form
\begin{equation}
  \label{eq:dialect_def}
    P(k_1,k_2,\dotsc, k_n | A) = \begin{cases}
     0 \text{ if } k_{j} = 0 \text{ and } K_j \in R_A,\\
    P(k_1|A) \dotsb P(k_i|A) P(k_{i+1}=1,\dotsc|A)\\ \text{ otherwise.}\\
    \end{cases}
\end{equation}
\end{definition}

If two messages really mean the same thing, then we may treat one, the other, or both as required.
This implies that dialects are ambiguous in a limited way,
but as we show in Theorems \ref{thm:min_decomp} and \ref{thm:minimal_count} in the Appendix Section \ref{sec:appendix} that this ambiguity is rather benign.
In short, bounds on the number of dialects and the dialects themselves can be obtained algorithmically.

\subsection{Conditionally independent mixtures with required messages}
\label{sec:conditionally_indep}

Consider a set of all messages $K$.
The \emph{power set} $2^K$ of $K$ consists of all subsets of $K$.
The power set is partially ordered by subset $\subseteq$ so that $(2^K,\subseteq)$ is a partially ordered set.
Equivalently, we can think of each subset of $K$ as a binary sequence of length $\#K$.
Equation \eqref{eq:dialect_def} can be thought of as defining a function $P(\cdot|A): 2^K \to [0,1]$ for each dialect $A$.
As an aside, the partial order consisting of those subsets of $K$ that occur for a dataset of files is a sub-order of the face order of the Dowker complex \cite{bjorner1995topological}.

From what was proved in \cite[Lem. 1]{Robinson_langsec2022},
we would then expect that once the required messages occur,
the probabilities defined by Equation \eqref{eq:dialect_def} decrease as more messages occur.
That is to say, the probabilities decrease as more of the $k_1,\dotsc,k_i$ take the value $1$.

\begin{lemma}
  \label{lem:struct}
  Suppose that within a dialect $A$, $P(k_j=1|A) < 1/2$ for every message $K_j$.
  The probability function defined by Equation \eqref{eq:dialect_def} can also be written as
  \begin{equation*}
    P(k_1, \dotsc, k_n | A ) = 1_{U_R} (k_1, \dotsc, k_n) g(k_1, \dotsc, k_n),
  \end{equation*}
  where
  \begin{enumerate}
    \item $R$ is the set of required messages for $A$,
    \item $U_{R} = \{B \in 2^K : R \subseteq B \}$ is the set of all message patterns containing $R$,
    \item $1_{U_R}$ is the indicator function on $U_R$, and
    \item $g: 2^K \to [0,1]$ is monotonic decreasing.
  \end{enumerate}
\end{lemma}
\begin{proof}
  Equation \eqref{eq:dialect_def} stipulates that if $k_{j} = 0$ for some $j \in \{i+1 \dotsc, n\}$, then $P(k_1, \dotsc, k_n | A )=0$.
  This means that the support of $P(\cdot | A)$ is contained within the support of the indicator function on the set
  \begin{equation*}
    R = \{K_{i+1}, \dotsc, K_n\}
  \end{equation*}
  of required messages.

  On the other hand, since $P(k_j=1|A) < 1/2$ for every message $K_j$, this implies that the probability decreases if we leave out a non-required message
  \begin{equation*}
    \begin{aligned}
      P(k_1, \dotsc, k_n | A ) &= P(k_1|A) \dotsb P(k_i|A) P(k_{i+1}=1,\dotsc|A) \\
      & < P(k_1|A) \dotsb P(k_{j-1}|A) P(k_{j+1}|A) \dotsb\\
      & \quad\quad  P(k_i|A) P(k_{i+1}=1,\dotsc|A) \\
      & < P(k_1, \dotsc, k_{j-1},k_{j+1} \dotsc | A ).
    \end{aligned}
  \end{equation*}
  Said another way, the probability is a monotonic decreasing function within $U_R$.
\end{proof}

If several dialects are present, the probability of message patterns being exhibited has a rather definite form.

\begin{corollary}
  \label{cor:mixture_struct}
  Suppose that within a dialect $A$, $P(k_j=1|A) < 1/2$ for all messages $k_j$.  
  Assuming dialects are disjoint, the joint distribution of messages over all files is then
  \begin{equation*}
    \label{eq:mixture_struct}
    \begin{aligned}
      P(k_1,k_2,\dotsc, k_n) &= \sum_{A} P(k_1,k_2,\dotsc, k_n | A) P(A) \\
      &= \sum_{A} 1_{U_{R_A}}(k_1, \dotsc, k_n) g_A(k_1, \dotsc, k_n),
    \end{aligned}
  \end{equation*}
  where $1_{U_{R_A}}$ and $g_A$ are the functions defined in the statement of Lemma \ref{lem:struct} associated to dialect $A$.
\end{corollary}

\begin{proposition}
  \label{prop:minimal_required}
  Suppose that within a dialect $A$, $P(k_j=1|A) < 1/2$ for all messages $k_j$.  
  Under the model given by Equation \eqref{eq:dialect_def},
  the support of a dialect (the set of message patterns where its probability is nonzero)
  has a unique minimal number of messages that occur, namely the required messages.
\end{proposition}
\begin{proof}
  According to Lemma \ref{lem:struct}, for a dialect $A$,
  \begin{equation*}
    P(k_1, \dotsc, k_n | A ) = 1_{U_R} (k_1, \dotsc, k_n) g(k_1, \dotsc, k_n)
  \end{equation*}
  for a monotonic decreasing $g$.
  This ensures that the support of the dialect is contained within $U_R$. 

  By way of contradiction, suppose that there were two minimal sets of messages that occur.
  This is equivalent to saying that there are at least two proper subsets $S_1 \subset R$ and $S_2 \subset R$ of the required messages $R$ for which $P(S_1|A)$ and $P(S_2|A)$ are both nonzero,
  yet $P(R|A)=0$.
  Notice that by construction,
  \begin{equation*}
    P(S_1|A) = g(S_1),
  \end{equation*}
  and
  \begin{equation*}
    P(R|A) = g(R).
  \end{equation*}
  We have just shown that $g(S_1) > g(R)$,
  yet this violates monotonicity and so is a contradiction.
\end{proof}

\section{Experimental results}
\label{sec:experimental}

Proposition \ref{prop:minimal_required} yields a decomposition of the joint message probability into dialects specified by minimal sets of required messages.
Assuming that these decompositions can be obtained---Section \ref{sec:methodology} describes an algorithm for constructing them---this
section discusses how these decompositions can partition a given set of files into semantically useful dialects for further exploration by other means.

\subsection{CSV}

The humble comma separated value (CSV) file format appears at first glance to be completely defined by its name.
It is a text file format for specifying tabular data, consisting of cells grouped into rows and columns.
Each row corresponds to a line in the file, delimited by one of a handful of line ending characters.
Each column is delimited by a comma character.
This simple characterization quickly goes awry as what constitutes a ``comma'' varies with language and text file encoding \cite{van2019wrangling}.
Moreover, since cells might contain delimiters for line endings or commas, some kind of quoting is required.
Again, quote characters vary with encoding.
Finally, because CSV files are often consumed by spreadsheet applications,
they can contain formulae or fragments of executable code that can interact in surprising ways \cite{poc19}.

To demonstrate the dialects that are present within a corpus of CSV files,
we drew a simple random sample of $3005$ files from the dataset described in \cite{van2019wrangling}.
To each of these files, we extracted a total of $33$ messages obtained by the {\tt CleverCSV} tool described in the same article.
In total, the messages consist of
\begin{itemize}
\item 14 delimiters,
\item 3 quote characters,
\item 3 escape characters, and
\item 13 distinct text encodings.
\end{itemize}
Figure \ref{fig:csv_matrix} shows a representation of the data obtained by this process.
Each row corresponds to a distinct message, and each column corresponds to a distinct file.
There are $29$ rows in Figure \ref{fig:csv_matrix} because $4$ messages ($1$ quote character, $2$ escape characters, and $1$ delimiter) did not occur on any file in our sample.

\begin{figure}
  \begin{center}
    \includegraphics[width=3in]{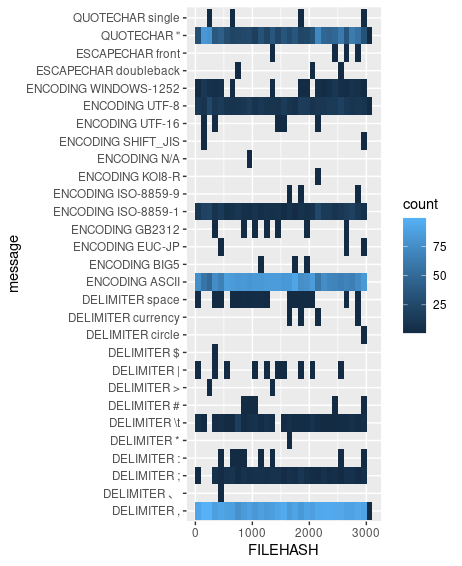}
    \caption{Summary of our CSV data: Rows correspond to possible messages, while columns correspond to files within our sample.  A gray cell indicates that the message did not occur for the corresponding file, whereas a shaded cell indicates that it occurred.  Due to limitations of horizontal resolution, files are aggregated into bins in which counts are displayed.}
    \label{fig:csv_matrix}
  \end{center}
\end{figure}

The horizontal stripes in Figure \ref{fig:csv_matrix} indicate that there are some messages that are very frequent (ASCII encoding and ASCII comma, for instance).
Table \ref{tab:csv_dowker} records the number of files exhibiting each message pattern.
Only message patterns for which $5$ or more files were present are shown.
Most CSV files (about $89\%$) correspond to text files delimited with commas, which is hopefully not too surprising.
ASCII text files delimited by ASCII commas correspond to about $70\%$ of the total sample.
About half of these ASCII files contain ASCII {\tt "} as a quotation character.
It is worth noting that the ASCII TAB character is often used as a field delimiter in Microsoft Excel.

\begin{table}
    \begin{center}
      \caption{File counts for the most common message patterns in the CSV data} 
      \label{tab:csv_dowker}
\begin{tabular}{|r|l|}
  \hline
  File count & Message pattern\\
  \hline
1417&	{\tt , ASCII}\\
682&	{\tt , " ASCII}\\
196&	{\tt , " UTF-8}\\
156&	{\tt , " ISO-8859-1}\\
119&	{\tt , UTF-8}\\
64&	{\tt , " WINDOWS-1252}\\
55&	{\tt ; ISO-8859-1}\\
51&	{\tt TAB ASCII}\\
48&	{\tt TAB ISO-8859-1}\\
45&	{\tt ; ASCII}\\
29&	{\tt ASCII}\\
18&	{\tt space ASCII}\\
18&	{\tt , ISO-8859-1}\\
11&	{\tt ; " ASCII}\\
10&	{\tt ; " UTF-8}\\
8&	{\tt | ASCII}\\
8&	{\tt , GB2312}\\
8&	{\tt ; " ISO-8859-1}\\
7&	{\tt : ASCII}\\
6&	{\tt , WINDOWS-1252}\\
\hline
\end{tabular}
    \end{center}
\end{table}

The perspective arising from Table \ref{tab:csv_dowker} is useful but does not do well for finding unusual dialects.
As described in \cite{Robinson_langsec2022},
a different summary of the file-message data is a partially ordered set in which each node corresponds to a message pattern,
and the order relation connects pairs of message patterns obtained by adding additional messages.
Figure \ref{fig:csv_dowker} shows the Hasse diagram of this partial order for our data.
In the figure, the size of each node is determined by the number of files exhibiting its corresponding message pattern.
Because the number of files in each message pattern varies tremendously, the sizes are scaled logarithmically.

It is immediately apparent that the graph consists of many disconnected components,
each of which correspond to at least one dialect.
Many of these components correspond to different text encodings.
As should be expected,
the ASCII-encoded files form the largest component.

\begin{figure*}
  \begin{center}
    \includegraphics[width=6in]{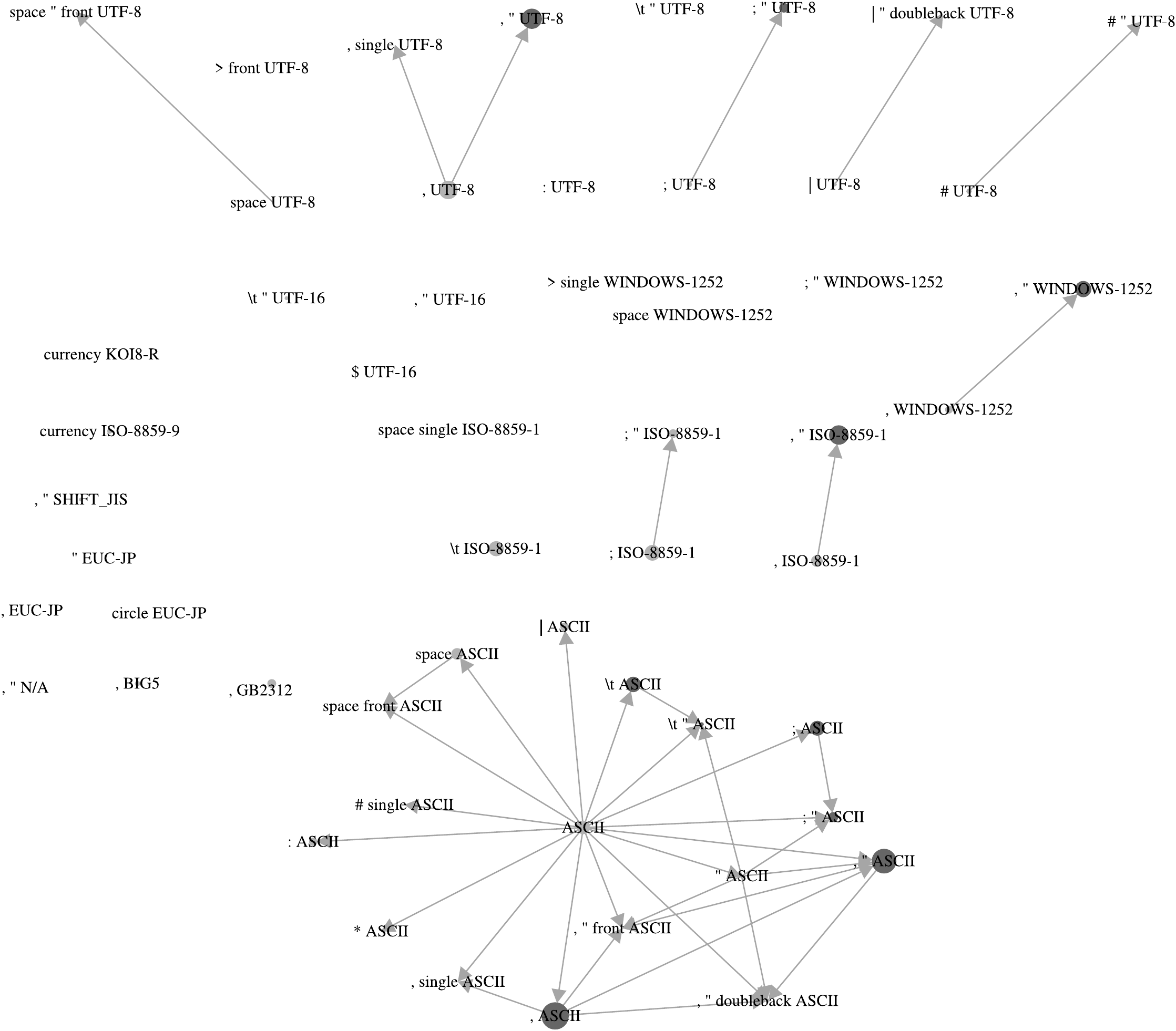}
    \caption{Partial order of message patterns.  Only those message patterns present in the data are shown.  The size of each node is logarithmic in the number of files exhibiting the corresponding message pattern.}
    \label{fig:csv_dowker}
  \end{center}
\end{figure*}

The statistical model described in Section \ref{sec:statistical_model} suggests that relationships between message patterns might account for their prevalence.
For instance, consider the message pattern {\tt TAB ASCII}.
Given that $51$ of these files occur, it should not be surprising that there are also some files that exhibit these two messages along with a quotation or escape character.
Moreover, the statistical model suggests if these other files are not too common, then we should consider them to be part of the same dialect.

There is some ambiguity in how the dialects are formed,
though Theorem \ref{thm:minimal_count} in Section \ref{sec:refinement} establishes that there is a unique set of dialects that are the largest.
Taking that as a given,
the resulting dialects we obtain are shown in Table \ref{tab:csv_dialect_output}.

\begin{table}
    \begin{center}
      \caption{Largest CSV candidate dialects, in order of discovery}
      \label{tab:csv_dialect_output}
\begin{tabular}{|r|l|}  
  \hline
  File count &  Required messages\\
  at root &\\
\hline
       1388 &{\tt , ASCII}\\
        119 &{\tt , UTF-8}\\         
         77 &{\tt , " UTF-8}\\
         18 &{\tt , ISO-8859-1}\\     
          7 &{\tt ; " UTF-8}\\       
         22 &{\tt TAB ASCII}\\        
        138 &{\tt , " ISO-8859-1}\\  
         48 &{\tt TAB ISO-8859-1}\\   
         58 &{\tt , " WINDOWS-1252}\\
         29 &{\tt ASCII}\\            
         55 &{\tt ; ISO-8859-1}\\     
         16 &{\tt ; ASCII}\\          
          8 &{\tt , GB2312}\\         
          6 &{\tt , WINDOWS-1252}\\
  \hline
\end{tabular}
    \end{center}
\end{table}

The ``File count at root'' column reports the number of files exhibiting the required messages \emph{and no others}.
According to Lemma \ref{lem:struct}, the number of files in the dialect associated to any other message pattern will not exceed this value.
In terms of file count, the top entries in Table \ref{tab:csv_dialect_output} are the comma-delimited text files with various text encodings.
The number of dialects with at least $5$ files is $14$, which is smaller than the $20$ message patterns with at least $5$ files,
so some consolidation of the data has occurred.
The {\tt , " ASCII} message pattern has been subsumed into {\tt , ASCII}, for instance.
The numbers of files in each have decreased somewhat because the statistical model expects a few files exhibiting message pattern {\tt , ASCII} to arise as part of the {\tt ASCII} dialect.
In short, the model expects that there are probably a few files that exhibited message pattern {\tt , ASCII} that are not actually CSV files---tabular data---but
are instead unstructured ASCII text files with commas.

\subsection{PDF}

The Portable Document File (PDF) format is defined by the ISO 32000-2 standard.
For this exercise, we used a sample of $10000$ files curated by the Test and Evaluation Team for the DARPA SafeDocs evaluation exercise 4.

Each file was processed through $11$ distinct parsers, run with various options.
A total of $3004$ Boolean messages were collected, as shown in Table \ref{tab:pdf_parsers}.
One message per parser is an exit code corresponding to the presence of an error.
The rest of the messages correspond to specific regular expressions (regexes) run against {\tt stderr} and {\tt stdout}, as explained previously in \cite{Robinson_langsec2022}.
Several of these messages were found to play an important role in identifying dialects, and appear in Table \ref{tab:pdf_messages}.

\begin{table}
  \begin{center}
    \caption{Messages associated to each PDF parser}
    \label{tab:pdf_parsers}
    \begin{tabular}{|l|r|}
      \hline
      Parser & Message count\\
      \hline
{\tt caradoc}&       253\\
{\tt hammer}&        65\\
{\tt mutool}&      796\\
{\tt origami}&      39\\
{\tt pdfium}&     21\\
{\tt pdfminer}&      62\\
{\tt pdftk}&       29\\
{\tt pdftools}&       10\\
{\tt poppler}&       995\\
{\tt qpdf}&          136\\
{\tt xpdf}&          598\\
\hline
Total & 3004\\
      \hline
    \end{tabular}
  \end{center}
\end{table}

\begin{table*}
    \begin{center}
      \caption{Sample messages in our PDF data relevant for candidate dialects}
      \label{tab:pdf_messages}
\begin{tabular}{|c|c|l|}
  \hline
  Message & parser & regex\\
  \hline
  69      &caradoc     &{\tt PDF error : Syntax error at offset \textbackslash d+ \textbackslash [0x[A-Fa-f\textbackslash d]+\textbackslash ] in file !}\\
  163     &caradoc     &{\tt PDF error : Syntax error at offset .* in file !}\\
  217     &caradoc     &{\tt PDF error : Lexing error : unexpected character : 0x[A-Fa-f\textbackslash d]+ at offset}...\\
  220     &caradoc     &{\tt PDF error : Lexing error : unexpected word at offset \textbackslash d+ \textbackslash [0x[A-Fa-f\textbackslash d]+\textbackslash }...\\
    250     &caradoc     &{\tt Warning : Flate\textbackslash/Zlib stream with appended newline in object .*}\\
   96,188,251     &caradoc     & Exit code meaning error \\
   \hline
 255     &hammer      &{\tt .*: no parse}\\
 258     &hammer      &{\tt (?:/[a-zA-Z\textbackslash d \textbackslash -]+)+/[A-Fa-f\textbackslash d]+: error after position \textbackslash d+ \textbackslash (0x[A-Fa-f\textbackslash d}...\\
 297     &hammer      &{\tt VIOLATION } ... {\tt No newline before 'endstream' }...\\
 308     &hammer      &{\tt VIOLATION } ... {\tt Missing endobj token \textbackslash (seve}...\\
 313     &hammer      &{\tt VIOLATION } ... {\tt No linefeed after 'stream' \textbackslash}...\\
 314     &hammer      &{\tt VIOLATION } ... {\tt Nonconformant WS at end of x}...\\
 316     &hammer      &Exit code meaning error \\
 \hline
 482     &mutool      &{\tt warning: line feed missing after stream begin marker \textbackslash (\textbackslash d+ \textbackslash d+ R\textbackslash )}\\
 720     &mutool      &{\tt warning: line feed missing after stream begin marker \textbackslash (\textbackslash d+ \textbackslash d+ R\textbackslash )}\\
 899     &mutool      &{\tt page (?:/[a-zA-Z\textbackslash d]+)+/[A-Fa-f\textbackslash ]+ \textbackslash d +}\\
 978     &mutool      &{\tt warning: line feed missing after stream begin marker \textbackslash (\textbackslash d+ \textbackslash d+ R\textbackslash )}\\
\hline
 1143    &origami     &{\tt .*Object shall end with 'endobj' statement.*}\\
 1153    &origami     & Exit code meaning error \\
 \hline
 2346    &qpdf        &{\tt WARNING: .*: expected endobj}\\
 2384    &qpdf        &{\tt WARNING: .*: stream keyword followed by carriage return only}\\
 \hline
 2889    &xpdf        &{\tt Syntax Warning.*: Substituting font '.*' for '.*'}\\
 3015    &xpdf        &{\tt non\_embedded\_font}\\
  \hline
\end{tabular}
    \end{center}
\end{table*}

After processing, we found $1658$ distinct message patterns, of which those with file count of at least $100$ are shown in Table \ref{tab:pdf_dowker}.
It should be noted that valid files often still produce numerous warnings and other output,
and so do not correspond to a small number of message patterns.

\begin{table}
    \begin{center}
      \caption{File counts for the most common message patterns in the PDF data} 
      \label{tab:pdf_dowker}
\begin{tabular}{|r|l|}
  \hline
  File count & Message pattern\\
  \hline
   3761& 250 251 899 1153                            \\                                                
    347& 251 297 899 1153                            \\                                                
    319& 250 251 899 1153 2888                       \\                                                
    186& 217 251 899 1153                            \\                                                
    117& 92 96 184 188 247 251 899 1153              \\                                                
    116& 200 250 251 899 1153                         \\                                               
    114& 69 96 163 188 220 251 255 258 297 308 313 ... \\
    107& 234 251 899 1153                  \\                                                          
    102& 228 251 899 1153                \\                                                            
    101& 7 96 104 188 217 251 899 1153 \\
\hline
\end{tabular}
    \end{center}
\end{table}

\begin{table}
    \begin{center}
      \caption{Largest PDF candidate dialects}
      \label{tab:pdf_dialect_output}
\begin{tabular}{|r|l|l|}  
  \hline
  File count &  Required messages & Interpretation\\
  at root &&\\
  \hline
   3684& 250 251 899 1153 & Compressed stream error\\
    270& 251 297 899 1153 & Missing/misplaced\\ && {\tt endstream} delimiter\\
    111& 69 96 163 188 220 251 255& Syntax error\\
    &258 297 308 313 ... & \\
    109& 217 251 899 1153 & Syntax error \\
  \hline
\end{tabular}
    \end{center}
\end{table}

Our proposed dialect decomposition algorithm compresses the message patterns into $4$ dialects with root file count greater than $100$,
as shown in Table \ref{tab:pdf_dialect_output}.
These four dialects have a fairly clear interpretation, as shown in the last column.
The latter two dialects appear to correspond to different kinds of syntax errors.

As a comparison, if we consider file counts of at least $25$, we found $10$ dialects with root file count greater than $25$ compared to $43$ message patterns.
In both cases, a format analyst need only consider about one-quarter as many dialects as overall message patterns.

\subsection{NITF}

The National Imagery Transmission Format (NITF) is used by many US Government entities to share geocoded imagery.
These files contain a complicated header and a data payload.
There are several variants of NITF in practice,
which has lead to some format divergence.

As part of the DARPA SafeDocs hackathon exercise 5,
the Test and Evaluation Team provided a set of $2626$ NITF files.
Against each of these files, $6$ parsers were run.
The output of {\tt stderr}, {\tt stdout} were collected along with the parser's return code.
Using the same process described in \cite{Robinson_langsec2022} for PDF, suitably modified for NITF,
a collection of regexes were run on this output to produce $136$ possible messages.
The breakdown of parsers and messages is shown in Figure \ref{tab:bae_nitf_parsers}.

\begin{table}
    \begin{center}
      \caption{Parsers used to process the NITF data}
      \label{tab:bae_nitf_parsers}
\begin{tabular}{|r|l|}
  \hline
  Parser     &     Message count\\
\hline
 {\tt afrl}    &  42\\
 {\tt codice}     &  28\\
 {\tt gdal}      &  36\\
 {\tt hammer\_nitf}  &  15\\
 {\tt kaitai}      &    6\\
 {\tt nitro}       &    9\\
 \hline
 Total & 136\\
 \hline
\end{tabular}
    \end{center}
\end{table}

The most common messages are shown in Table \ref{tab:bae_nitf_messages}.
Several regular expressions matched more than 50\% of the time.
Since this violates the assumptions on message probability in Section \ref{sec:conditionally_indep},
it is more informative to use the absence of a match instead.

\begin{table*}
    \begin{center}
      \caption{Most common messages in the NITF data}
      \label{tab:bae_nitf_messages}
\begin{tabular}{|c|r|c|l|}
  \hline
  Message & File count & parser & regex\\
  \hline
59 & 1051 & {\tt codice}       & Absence of {\tt Parse error\textbackslash n }\\
102& 1039 & {\tt gdal}         & Absence of {\tt gdalinfo failed \textbackslash - unable to open '.*'\textbackslash . }\\
107& 1038 & {\tt hammer\_nitf} & Absence of errors in exit code\\
71 & 1029 & {\tt gdal}         & Absence of errors in exit code\\
1  &  825 & {\tt afrl}         & Absence of errors in exit code\\
37 &  812 & {\tt afrl}         & {\tt Error reading, read returned .*\textbackslash . \textbackslash (start = .*,} ...\\
94 &  527 & {\tt gdal}         & {\tt ERROR \textbackslash d+: Not enough bytes to read segment info }\\
108&  470 & {\tt hammer\_nitf} & {\tt /[a-zA-Z\textbackslash d \_\textbackslash \textbackslash .\textbackslash -\textbackslash (\textbackslash ):/,+]+\textbackslash .[a-zA-Z\textbackslash d]+: no parse }\\
113&  420 & {\tt hammer\_nitf} & {\tt VIOLATION } ... {\tt Invalid file length in header \textbackslash (severity=\textbackslash d+\textbackslash ) }\\
21 &  394 & {\tt afrl}         & {\tt Error reading header.* }\\
12 &  308 & {\tt afrl}         & {\tt user defined data length = \textbackslash d+ }\\
103&  241 & {\tt gdal}         & {\tt gdal ERROR .*: NITF Header Length \textbackslash (.*\textbackslash ) seems}...\\
119&  241 & {\tt hammer\_nitf} & {\tt VIOLATION } ... {\tt Invalid number of graph segments \textbackslash (severity=\textbackslash d+\textbackslash ) }\\
99 &  227 & {\tt gdal}         & {\tt Warning \textbackslash d+:} ... {\tt appears to be an NITF file, but no image } ... \\
  \hline
\end{tabular}
    \end{center}
    \end{table*}

In our data, $103$ of the messages are extant.
There are $20$ distinct message patterns with file count at least $25$,
which are shown in Table \ref{tab:bae_nitf_dowker}.

\begin{table}
    \begin{center}
      \caption{File counts for the most common message patterns in the NITF data} 
      \label{tab:bae_nitf_dowker}
\begin{tabular}{|r|l|}
  \hline
  File count & Message pattern\\
  \hline
357	&1 59 71 102 107\\
151	&1 12 59 71 102 107\\
100	&1 59 71 99 102 107 122\\
70	&14 23 94\\
70	&94\\
66	&21 37 81 113\\
59	&1 59 71 99 102 107\\
56	&103 113\\
50	&15 37 86\\
48	&21 37 94 113\\
47	&33 37 103 113\\
44	&22 37 76 108 119\\
43	&94 111\\
42	&21 37 103 113\\
30	&21 37 94 119\\
29	&17 59 71 102 107\\
29	&21 37 76 113\\
29	&7 59 71 102 107\\
28	&12 34 37 40 94\\
26	&2 12 59 79 82 107\\
  \hline
\end{tabular}
    \end{center}
\end{table}

The decomposition of the data into dialects is summarized in Table \ref{tab:bae_nitf_dialect_output}.
The ``File count at root'' column reports the number of files exhibiting the required messages \emph{and no others}.
According to Lemma \ref{lem:struct}, the number of files in the dialect associated to any other message pattern will not exceed this value.

There are $12$ dialects with required messages having a file count at least $25$, as shown in Table \ref{tab:bae_nitf_dialect_output}.
This is an improvement over the $20$ message patterns that Table \ref{tab:bae_nitf_dowker} presents.
Referring back to Table \ref{tab:bae_nitf_messages},
all but the two most common dialects shown contains some kind of parser \emph{error} message.

\begin{table}
    \begin{center}
      \caption{Largest NITF candidate dialects}
      \label{tab:bae_nitf_dialect_output}
\begin{tabular}{|r|l|l|}  
  \hline
  File count &  Required messages&Interpretation\\
  at root && \\
  \hline
352	&1 59 71 102 107 & Valid files\\
93	&1 59 71 99 102 107 122 & Corrupted data payload\\
70	&94& Read access error\\
60	&14 23 94&Read access error\\
54	&103 113&Corrupted header length\\
49	&15 37 86&Read access error\\
43	&21 37 81 113&Corrupted header length\\
41	&21 37 94 113&Corrupted header length\\
27	&22 37 76 108 119&Read access error\\
26	&21 37 94 119&Corrupted header\\
26	&2 12 59 79 82 107&Valid but unsupported version\\
25	&21 37 76 113&Corrupted header length\\
  \hline
\end{tabular}
    \end{center}
\end{table}

While not every possible message corresponds to violations of length fields,
most of the dialects shown in Table \ref{tab:bae_nitf_dialect_output} correspond to a corrupted length field within the NITF header.
Because the NITF specification permits random access to the data payload,
length field corruption explains the presence of messages $37$ and $94$,
which indicate reading beyond the end of the file.
What ultimately distinguishes the dialects is which length fields were corrupted.
Due to differences in how the parsers operate, different fields are collected at different points in the parse by different parsers.

Further examination of the files in the second-to-last dialect in Table \ref{tab:bae_nitf_dialect_output}, the one that contains message $2$,
reveals that these files were for a version not supported by the {\tt afrl} parser.

\section{Detailed methodology}
\label{sec:methodology}

Our goal is to find candidate dialects from the probabilities $P(k_1,k_2,\dotsc, k_n)$ of each message pattern.
We show in Proposition \ref{prop:decomp} that there are decompositions of these data into a mixture of disjoint dialects of the form postulated in Corollary \ref{cor:mixture_struct}.
Furthermore, there is a greedy algorithm for finding such a decomposition.

Although our probabilistic model is posed over the power set of messages partially ordered by subset $(2^K,\subseteq)$,
We will instead establish results for an arbitrary partially ordered set $P$.
This can yield a substantial savings in memory usage and runtime of any algorithms working on the data,
because we need only consider those message patterns that are actually present in the data.
Moreover, at the level of generality used, our data could be formatted as probabilities taking values between $0$ and $1$ or as counts of files.

This section establishes an algorithm that constructs candidate dialect decompositions from the data.
As suggested in Section \ref{sec:conditionally_indep},
the data are formatted as a function $f: P \to [0,\infty)$ from a partially ordered set $(P,\le)$ to the nonnegative real numbers.
  The key insight is Definition \ref{def:monotonic_decomp},
  which generalizes Corollary \ref{cor:mixture_struct} to handle all possible candidate dialect decompositions.

  The starting point is to relax from independent mixtures (Corollary \ref{cor:mixture_struct}) to a decomposition into monotonic functions.

  \begin{definition}
  \label{def:monotonic_decomp}
  Suppose that $f: P \to [0,\infty)$ is a function from a finite partially ordered set $(P,\le)$ to the nonnegative real numbers.
  A \emph{monotonic decomposition} expresses $f$ as a sum of functions 
  \begin{equation}
    \label{eq:monotonic_decomp}
    f(x) = \sum_{k=1}^{N} 1_{U_{y_k}}(x) g_k(x)
  \end{equation}
  where $g_k: P \to [0,\infty)$ is monotonic decreasing, and $U_{y_k} = \{x \in P : x \ge y_k \}$ is an upwardly closed set.

  Duplicate $y_k$ and/or $g_k$ are permissible in a monotonic decomposition.
\end{definition}

Monotonic decompositions generalize the probabilistic model posited in Section \ref{sec:statistical_model}.

\begin{proposition}
  \label{prop:dialect_monotonic_decomp}
  Suppose that there is a set of messages for which the probability each message pattern in each dialect is given by Equation \eqref{eq:dialect_def}.
  Assume that dialects are disjoint so that Corollary \ref{cor:mixture_struct} applies.
  Then Equation \eqref{eq:mixture_struct} is a monotonic decomposition of the joint probability distribution over all message patterns.
\end{proposition}
\begin{proof}
  According to Lemma \ref{lem:struct}, each dialect corresponds to a term of the form $1_{U_y} g(x)$ where $g$ is a monotonic decreasing function.
  According to Corollary, \ref{cor:mixture_struct},
  the formula for the joint probability distribution for all messages patterns is of precisely the same form as required
  by Equation \eqref{eq:monotonic_decomp} in Definition \ref{def:monotonic_decomp}.
\end{proof}

Our experimental results follow from a constructive proof of the following.

\begin{proposition}
  \label{prop:decomp}
  Every function $f: P \to [0,\infty)$ from a finite partially ordered set $(P,\le)$ to the nonnegative real numbers has a monotonic decomposition.
\end{proposition}

It is not true that the decompositions proposed in Proposition \ref{prop:decomp} are unique,
especially because there are several arbitrary choices that factor into its otherwise-constructive proof.

\begin{example}
  \label{eg:nonunique}
Consider the partially ordered set given by
\begin{equation*}
  \xymatrix{
    &D&\\
    B\ar[ur]&&C\ar[ul]\\
    &A\ar[ur]\ar[ul]&\\
    }
\end{equation*}
with the function $f$ given by
\begin{equation*}
  f(A):= 0, \;  f(B):= 4, \;  f(C):= 4, \text{ and }  f(D):= 5.
\end{equation*}
This function can be written as the sum
\begin{equation*}
  f = 1_{U_B} g_1 + 1_{U_C} g_2,
\end{equation*}
where
\begin{equation*}
  g_1(A):= 5, \;  g_1(B):= 4, \;  g_1(C):= 4, \text{ and }  g_1(D):= 2,
\end{equation*}
and
\begin{equation*}
  g_2(A):= 5, \;  g_2(B):= 4, \;  g_2(C):= 4, \text{ and }  g_2(D):= 3.
\end{equation*}
Evidently, both $g_1$ and $g_2$ are monotonic decreasing.
However, it is also true that
\begin{equation*}
  f = 1_{U_B} g_2 + 1_{U_C} g_1,
\end{equation*}
which contradicts the uniqueness of the decomposition, since exactly the same two open sets are used.
\end{example}

We will prove Proposition \ref{prop:decomp} constructively (algorithmically after certain choices are made) later in this section,
after establishing two Lemmas as tools.

\begin{lemma}
  \label{lem:join_semilattice}
  Suppose that $f: P \to [0,\infty)$ is a function from a partially ordered set $(P,\le)$ to the nonnegative real numbers.
    If $g_1,g_2: P \to [0,\infty)$ are two functions satisfying
      \begin{enumerate}
      \item Both $g_i$ are \emph{monotonic decreasing}: if $x \le y$ are two elements of $P$, then $g_i(x) \ge g_i(y)$, and
      \item $g_i(x) \le f(x)$ for all $x \in P$ and both $g_i$.
      \end{enumerate}
      Then the function
      \begin{equation*}
        h(x) := \max\{g_1(x), g_2(x)\}
      \end{equation*}
      satisfies the same two conditions.
\end{lemma}
\begin{proof}
  The fact that $h(x) \le f(x)$ for all $x\in P$ follows immediately from the fact that $h(x)$ is equal to either $g_1(x)$, or $g_2(x)$, or both.
  Now suppose that $x \le y$ are two elements of $P$.
  Without loss of generality, suppose that $h(x) = g_1(x)$.
  This means that $g_1(x) \ge g_2(x)$, and by assumption $g_2(x) \ge g_2(y)$.
  By transitivity, this means that $h(x) = g_1(x) \ge g_2(y)$.
  By assumption, we have that $h(x) = g_1(x) \ge g_1(y)$ as well.
  Since $h(y)$ is equal to the larger of $g_1(y)$ and $g_2(y)$, this means that $h(x)$ is larger than $h(y)$.
\end{proof}

\begin{figure}
  \begin{center}
    \includegraphics[width=2.5in]{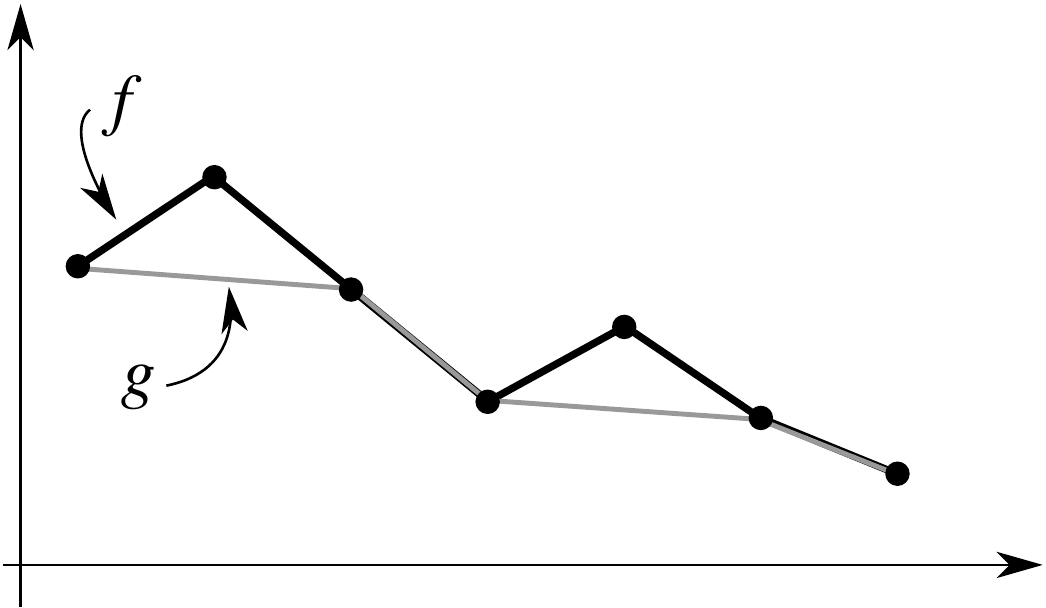}
    \caption{An example of the function $g$ guaranteed by Lemma \ref{lem:max_monotonic_lower_bound}}
    \label{fig:max_monotonic_lower_bound}
  \end{center}
\end{figure}

\begin{lemma}
  \label{lem:max_monotonic_lower_bound}
  Suppose that $f: P \to [0,\infty)$ is a function from a finite partially ordered set $(P,\le)$ to the nonnegative real numbers.
    The set of monotonic decreasing functions $g: P \to [0,\infty)$ such that $g(x) \le f(x)$ for all $x \in P$ has a unique maximal element.
      (See Figure \ref{fig:max_monotonic_lower_bound}.)
\end{lemma}
\begin{proof}
  All we need to show is that the set in question is nonempty, since Lemma \ref{lem:join_semilattice} does the rest.
  This is easy because the zero function is in the set.
\end{proof}

There is a greedy, recursive algorithm that constructs the function $g$ guaranteed by Lemma \ref{lem:max_monotonic_lower_bound}.

\begin{enumerate}
\item Start with the set of minimal elements $M_0$ of $P$.
\item We can without any trouble define $g(m) := f(m)$ for every $m \in M$, as this is evidently maximal in all cases.
\item In preparation for the recursive step, let $L_0 := M_0$.
  
\item For there recursive step, assume $g$ is already defined on some lower-closed subset $L_k$ of $P$.
  \begin{enumerate}
  \item Consider the set $M_k$ of minimal elements of $P \setminus L_k$.
  \item For each $m \in M_k$, define $g(m) := \min\{f(m)\} \cup \{g(x) : x < m \}$, noting that every $x$ in the latter set is an element of $L_k$ so $g(x)$ is well-defined.
    Defining $g$ in this way ensures that it is upper bounded by $f$, is monotonic decreasing, yet is otherwise maximal.
  \item In preparation for the next recursive step, let $L_{k+1} := L_k \cup M_k$.
  \end{enumerate}
\end{enumerate}

\begin{figure}
  \begin{center}
    \includegraphics[width=2.5in]{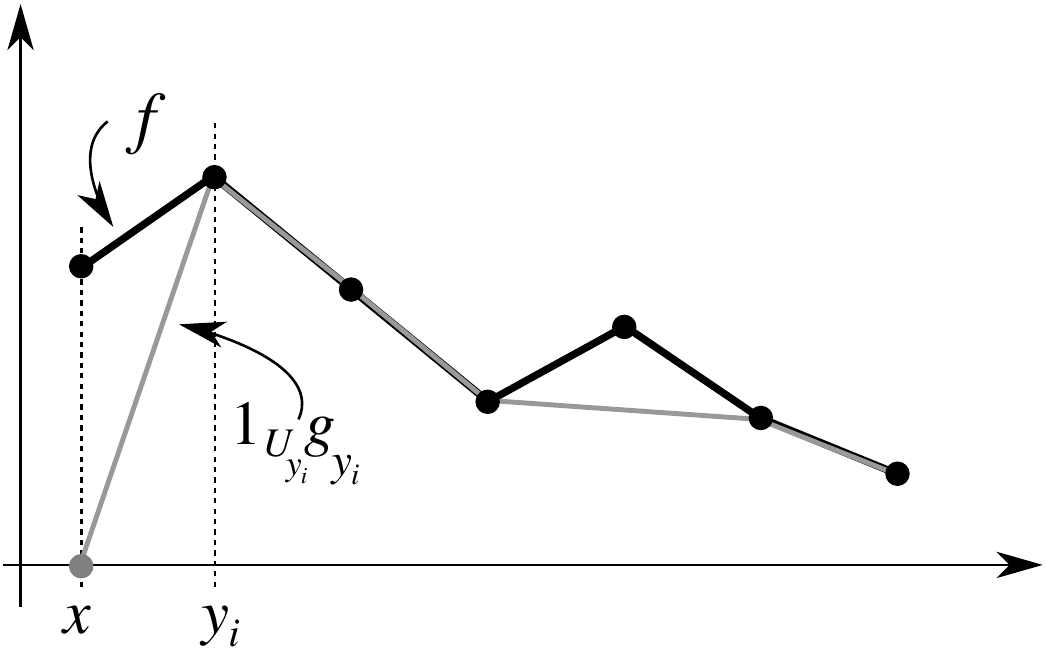}
    \caption{Construction of the function $g_{y_i}$ in the proof of Proposition \ref{prop:decomp}}
    \label{fig:decomp}
  \end{center}
\end{figure}

These maximal monotonic decreasing functions can be used to decompose an arbitrary function into a sum of monotonic decreasing functions whose domains are restricted appropriately.
This allows us to prove Proposition \ref{prop:decomp}.

\begin{proof}(of Proposition \ref{prop:decomp})
  Proceed by induction on the number of places where $f$ fails to be monotonic decreasing.

  \begin{itemize}
  \item Base case: $f$ is monotonic decreasing.
    If there is one minimal element, $y$ of $P$, then we merely take $y_1:=y$, and let $g_1 := f$.
    However, if there is more than one minimal element, then things become annoyingly non-unique.
    This can be resolved in various ways, for instance using the following tie-breaking procedure.
    Let $y_1, \dotsc, y_m$ be an arbitrary ordering of the minimal elements of $P$.
    Since they are all minimal, they are mutually incomparable elements of $P$.  Define
  \begin{equation*}
    g_i(x) := 1_{A_i}(x) f(x),
  \end{equation*}
  where
  \begin{equation*}
    A_i = U_{y_i} \setminus \bigcup_{j=1}^{i-1} U_{y_j}.
  \end{equation*}
  Notice that since each $1_{A_i}$ is monotonic decreasing---it is $1$ on $y_i$,
  but eventually drops to $0$ on sufficiently large elements of $P$---the resulting $g_i$ functions are also monotonic decreasing.
  Moreover, by construction each element $x$ of $P$ is an element of exactly one $A_i$ set.
  Therefore, the decomposition formula for $f$ holds.

\item Induction case: Suppose that $f$ is not monotonic decreasing at $k$ elements $\{y_1, \dotsc, y_k\}$ of $P$.
  For each of these elements $y_i$, there is an $x$ in $P$ with $x \le y_i$ but $f(x) < f(y_i)$.
  At least one of these $y_i$ is minimal among the set $\{y_1, \dotsc, y_k\}$, which means that for each $j \not= i$,
  either $y_i < y_j$ or $y_i$ and $y_j$ are incomparable.
  Use Lemma \ref{lem:max_monotonic_lower_bound} to construct a maximal monotonic decreasing function $g_{y_i}$ on $U_{y_i}$ that is bounded above by $f$.
  We will show that the residual function $f-1_{U_{y_i}} g_{y_i}$ fails to be monotonic at not more than $k-1$ elements of $P$.  

  We assumed that there was an $x$ such that $f(x) < f(y_i)$.
  Due to the hypothesis that $g_{y_i}$ is bounded above by $f$,
  this means that $g_{y_i}(y_i) \le f(y_i)$.
  On the other hand, given that Lemma \ref{lem:max_monotonic_lower_bound} asserts a maximal such $g_{y_i}$ exists on $U_{y_i}$,
  we must conclude that $g_{y_i}(y_i) = f(y_i)$,
  as shown in Figure \ref{fig:decomp}, since that value has no further impact on the monotonicity of $g_{y_i}$.
  Therefore, the residual function $f-1_{U_{y_i}} g_{y_i}$ takes the value $0$ on $y_i$, and therefore automatically satisfies
  \begin{equation*}
    \begin{aligned}
      f(x) &= f(x) - 1_{U_{y_i}}(x) g_{y_i}(x) \\
           &\ge 0 =  f(y_i) - 1_{U_{y_i}}(y_i) g_{y_i}(y_i).
    \end{aligned}
  \end{equation*}
  Therefore, at least one violation of monotonicity in $f$ is not present in the residual.

  Let us establish that no new violations of monotonicity occur in the residual.
  Suppose that $x \le z$ are two elements of $P$ for which $f(x) \ge f(z)$.
  If both elements are outside $U_{y_i}$, then the residual is unchanged from $f$ on these two elements.
  If $x$ is outside $U_{y_i}$ but $z \in U_{y_i}$, then
  \begin{equation*}
    \begin{aligned}
      f(x) - 1_{U_{y_i}}(x) g_{y_i}(x) &= f(x) \\
      & \ge f(z)\\
      & \ge f(z) - 1_{U_{y_i}}(z) g_{y_i}(z),
    \end{aligned}
  \end{equation*}
  since $g_{y_i}$ is nonnegative by construction.
  Finally, assume that both $x$ (and therefore $z$) are elements of $U_{y_i}$ and that $f(x) \ge f(z)$.
  While an arbitrary monotonic decreasing function $h$ bounded above by $f$ might not result in $f(x) - h(x) \ge f(z) - h(z)$,
  this cannot happen with $g_{y_i}$ due to its maximality.
  We establish this by way of contradiction; suppose that
  \begin{equation*}
    f(x) - g_{y_i}(x) < f(z) - g_{y_i}(z).
  \end{equation*}
  Rearranging this inequality yields
  \begin{equation*}
  0 \le f(x) - f(z) < g_{y_i}(x) - g_{y_i}(z),
  \end{equation*}
  which means that there is a monotonic decreasing $h$ with $h(x) = g_{y_i}(x)$ and $g_{y_i}(z) < h(z) \le f(z)$,
  contradicting the maximality of $g_{y_i}$.
  We have therefore established that the residual $f-1_{U_{y_i}}g_{y_i}$ has strictly fewer violations of monotonicity than $f$.
  \end{itemize}
\end{proof}

Unfortunately, the decomposition proposed in the proof of Proposition \ref{prop:decomp} is not unique.
This means that there are sometimes different possible choices of dialects that result in the same message probabilities.
In the Appendix Section \ref{sec:appendix}, Theorems \ref{thm:min_decomp} and \ref{thm:minimal_count} show that the decomposition constructed by the algorithm nevertheless yields a bound on the number and structure of dialects.
In particular, the sets of required messages found by the procedure correspond to those of some true dialects,
though there may be other dialects that remain to be found.

\section{Conclusion}

This paper presented a novel statistically-based method for partitioning sets of files into format dialects based upon their behaviors when parsed.
Theoretically, our method yields the coarsest such partition that could be consistent with the statistical model.
This means that a format analyst can begin their analysis with a minimal number of dialects.
In practice, an analyst needs to consider about half the number of dialects as distinct message patterns.
Intuitively, this considerably reduces their cognitive load when studying a complex format.

\section*{Acknowledgments}

The authors would like to thank the SafeDocs test and evaluation team,
including NASA (National Aeronautics and Space Administration) Jet Propulsion Laboratory, California Institute of Technology and the PDF Association, Inc., for providing the test data.
The authors would like to thank Cory Anderson for the initial processing of the files into sets of messages.

This material is based upon work supported by the Defense Advanced Research Projects Agency (DARPA) SafeDocs program under contract HR001119C0072.
Any opinions, findings and conclusions or recommendations expressed in this material are those of the author and do not necessarily reflect the views of DARPA.

\section*{Conflict of interest}
The authors state that there is no conflict of interest.

\bibliographystyle{IEEEtran}
\bibliography{dowkersegment_bib}

\begin{thebibliography}{10}
\providecommand{\url}[1]{#1}
\csname url@samestyle\endcsname
\providecommand{\newblock}{\relax}
\providecommand{\bibinfo}[2]{#2}
\providecommand{\BIBentrySTDinterwordspacing}{\spaceskip=0pt\relax}
\providecommand{\BIBentryALTinterwordstretchfactor}{4}
\providecommand{\BIBentryALTinterwordspacing}{\spaceskip=\fontdimen2\font plus
\BIBentryALTinterwordstretchfactor\fontdimen3\font minus
  \fontdimen4\font\relax}
\providecommand{\BIBforeignlanguage}[2]{{%
\expandafter\ifx\csname l@#1\endcsname\relax
\typeout{** WARNING: IEEEtran.bst: No hyphenation pattern has been}%
\typeout{** loaded for the language `#1'. Using the pattern for}%
\typeout{** the default language instead.}%
\else
\language=\csname l@#1\endcsname
\fi
#2}}
\providecommand{\BIBdecl}{\relax}
\BIBdecl

\bibitem{Robinson_langsec2022}
\BIBentryALTinterwordspacing
M.~Robinson, L.~W. Li, C.~Anderson, and S.~Huntsman, ``Statistical detection of
  format dialects using the weighted {Dowker} complex,'' in \emph{2022 IEEE
  Security and Privacy Workshops (SPW)}, 2022, pp. 98--112. [Online].
  Available: \url{http://arxiv.org/pdf/2201.08267}
\BIBentrySTDinterwordspacing

\bibitem{Scofield_2017}
D.~Scofield, C.~Miles, and S.~Kuhn, ``Fast model learning for the detection of
  malicious digital documents,'' in \emph{SSPREW-7}, December 2017.

\bibitem{Marin_2011}
J.~Marin, K.~Mengersen, and C.~P. Robert, ``Bayesian modelling and inference on
  mixtures of distributions,'' in \emph{Essential Bayesian models. Handbook of
  statistics: Bayesian thinking - modeling and computation. Vol. 25}, D.~Dey
  and C.~Rao, Eds.\hskip 1em plus 0.5em minus 0.4em\relax Elsevier, 2011).

\bibitem{McLachlan_2000}
G.~McLachlan and D.~Peel, \emph{Finite Mixture Models}.\hskip 1em plus 0.5em
  minus 0.4em\relax Wiley, 2000.

\bibitem{Robinson_langsec2021}
\BIBentryALTinterwordspacing
M.~Robinson, ``Looking for non-compliant documents using error messages from
  multiple parsers,'' in \emph{2021 IEEE Security and Privacy Workshops (SPW)},
  2021, pp. 184--193. [Online]. Available:
  \url{https://doi.org/10.1109/SPW53761.2021.00032}
\BIBentrySTDinterwordspacing

\bibitem{Ambrose_2020}
\BIBentryALTinterwordspacing
K.~Ambrose, S.~Huntsman, M.~Robinson, and M.~Yutin, ``Topological differential
  testing, {\tt arxiv:2003.00976},'' 2020. [Online]. Available:
  \url{https://arxiv.org/abs/2003.00976}
\BIBentrySTDinterwordspacing

\bibitem{belaoued2015real}
M.~Belaoued and S.~Mazouzi, ``A real-time {PE}-malware detection system based
  on chi-square test and {PE}-file features,'' in \emph{IFIP International
  Conference on Computer Science and its Applications}.\hskip 1em plus 0.5em
  minus 0.4em\relax Springer, 2015, pp. 416--425.

\bibitem{al2018ransomware}
B.~A.~S. Al-rimy, M.~A. Maarof, and S.~Z.~M. Shaid, ``Ransomware threat success
  factors, taxonomy, and countermeasures: A survey and research directions,''
  \emph{Computers \& Security}, vol.~74, pp. 144--166, 2018.

\bibitem{8685181}
\BIBentryALTinterwordspacing
S.~D. {S.L} and J.~{CD}, ``Windows malware detector using convolutional neural
  network based on visualization images,'' \emph{IEEE Transactions on Emerging
  Topics in Computing}, pp. 1--1, 2019. [Online]. Available:
  \url{https://doi.org/10.1109/TETC.2019.2910086}
\BIBentrySTDinterwordspacing

\bibitem{ALAZAB201591}
\BIBentryALTinterwordspacing
M.~Alazab, ``Profiling and classifying the behavior of malicious codes,''
  \emph{Journal of Systems and Software}, vol. 100, pp. 91 -- 102, 2015.
  [Online]. Available:
  \url{http://www.sciencedirect.com/science/article/pii/S0164121214002283}
\BIBentrySTDinterwordspacing

\bibitem{demme2013feasibility}
J.~Demme, M.~Maycock, J.~Schmitz, A.~Tang, A.~Waksman, S.~Sethumadhavan, and
  S.~Stolfo, ``On the feasibility of online malware detection with performance
  counters,'' \emph{ACM SIGARCH Computer Architecture News}, vol.~41, no.~3,
  pp. 559--570, 2013.

\bibitem{van2019wrangling}
\BIBentryALTinterwordspacing
G.~J.~J. {van den Burg}, A.~Naz{\'a}bal, and C.~Sutton, ``Wrangling messy {CSV}
  files by detecting row and type patterns,'' \emph{Data Mining and Knowledge
  Discovery}, vol.~33, no.~6, pp. 1799--1820, 2019. [Online]. Available:
  \url{https://doi.org/10.1007/s10618-019-00646-y}
\BIBentrySTDinterwordspacing

\bibitem{bjorner1995topological}
A.~Bj{\"o}rner, ``Topological methods,'' \emph{Handbook of combinatorics},
  vol.~2, pp. 1819--1872, 1995.

\bibitem{Robinson_relations}
\BIBentryALTinterwordspacing
M.~Robinson, ``Cosheaf representations of relations and {Dowker} complexes,''
  \emph{Journal of Applied and Computational Topology}, 2021. [Online].
  Available: \url{https://doi.org/10.1007/s41468-021-00078-y}
\BIBentrySTDinterwordspacing

\bibitem{brun2023rectangle}
M.~Brun and L.~M. Salbu, ``The rectangle complex of a relation,''
  \emph{Mediterranean Journal of Mathematics}, vol.~20, no.~1, pp. 1--8, 2023.

\bibitem{poc19}
\BIBentryALTinterwordspacing
J.~Dileo, ``{CSV} injection, {RFC5322},'' March 2019. [Online]. Available:
  \url{https://unpack.debug.su/pocorgtfo/pocorgtfo19.pdf}
\BIBentrySTDinterwordspacing

\bibitem{Kime_1974}
\BIBentryALTinterwordspacing
C.~R. Kime, R.~P. Batni, and J.~D. Russell, ``An efficient algorithm for
  finding an irredundant set cover,'' \emph{J. ACM}, vol.~21, no.~3, p.
  351–355, jul 1974. [Online]. Available:
  \url{https://doi.org/10.1145/321832.321833}
\BIBentrySTDinterwordspacing

\bibitem{minato1993fast}
S.-i. Minato, ``Fast generation of prime-irredundant covers from binary
  decision diagrams,'' \emph{IEICE transactions on fundamentals of electronics,
  communications and computer sciences}, vol.~76, no.~6, pp. 967--973, 1993.

\end{thebibliography}

\section{Appendix}
\label{sec:appendix}

Although Proposition \ref{prop:decomp} constructs a monotonic decomposition of a functions,
nonuniqueness can impede its practical utility.
Theorem \ref{thm:min_decomp} in Section \ref{sec:bounds} asserts that the decomposition constructed in the proof of Proposition \ref{prop:decomp} is minimal
in the sense that it cannot be decomposed further.
Furthermore, Theorem \ref{thm:minimal_count} asserts that the decomposition so constructed finds an unambiguous lower bound on the true number of dialects.

\subsection{Bounding the structure of dialects}
\label{sec:bounds}

An effective way to handle the ambiguity present among possible dialect decompositions is to simply embrace it.
Some decompositions are evidently finer, in that they split the set of files into smaller dialects.
This situation is easily characterized by the notion of \emph{refinement}.

\begin{definition}
Suppose that a function $f: P \to [0,\infty)$ from a finite partially ordered set $(P,\le)$ has multiple monotonic decompositions.
  We will say that the monotonic decomposition
  \begin{equation*}
    f(x) = \sum_{k=1}^{N} 1_{U_{z_k}}(x) h_k(x)
  \end{equation*}
  \emph{refines}
  the monotonic decomposition
  \begin{equation*}
    f(x) = \sum_{j=1}^{M} 1_{U_{y_j}}(x) g_j(x)
  \end{equation*}
  if for every $k = 1, \cdots, N$, there is a $j$ such that
  \begin{enumerate}
  \item $U_{z_k} \subseteq U_{y_j}$, and
  \item $h_k(x) \le g_j(x)$ for every $x \in U_{z_k}$.
  \end{enumerate}
\end{definition}

\begin{lemma}
  \label{lem:decomp_order}
  The set of monotonic decompositions of a function $f: P \to [0,\infty)$ from a finite partially ordered set $(P,\le)$
    is itself a preordered set under refinement.
\end{lemma}
\begin{proof}
  This is merely a straightforward verification of the axioms.  Suppose that we have three monotonic decompositions of $f$,
  \begin{equation*}
    \begin{aligned}
      f(x) &= \sum_{k=1}^{N} 1_{U_{z_k}}(x) h_k(x) \\
      & = \sum_{j=1}^{M} 1_{U_{y_j}}(x) g_j(x) \\
      &= \sum_{\ell=1}^{P} 1_{U_{w_\ell}}(x) p_{\ell}(x).
    \end{aligned}
  \end{equation*}
  \begin{itemize}
  \item Reflexivity is trivial; both $\subseteq$ for sets and $\le$ for functions are reflexive.

  \item For transitivity, suppose that the first monotonic decomposition refines the second,
    and that the second refines the third.
    For every $k=1, \cdots, N$, there is a $j$ such that $U_{z_k} \subseteq U_{y_j}$.
    Yet there is also an $\ell$ such that $U_{y_j} \subseteq U_{w_\ell}$.
    Thus $U_{z_k} \subseteq U_{w_\ell}$.
    For exactly these same indices, we have that
    \begin{equation*}
      h_k(x) \le g_j(x) \le p_{\ell}(x)
    \end{equation*}
    for all $x \in U_{z_k}$.
    Hence, the first monotonic decomposition refines the third.
  \end{itemize}
\end{proof}

We seek monotonic decompositions that are \emph{minimally refined},
there is no other monotonic decomposition refined by it.
Because of cases like Example \ref{eg:nonunique}, minimally refined monotonic decompositions are not unique.
In Example \ref{eg:nonunique}, both decompositions are minimally refined and neither refines the other.
Nevertheless, we have the following.

\begin{theorem}
  \label{thm:min_decomp}
  The procedure defined in the proof of Proposition \ref{prop:decomp} yields a minimally refined monotonic decomposition for an arbitrary nonnegative function $f$ on a finite partially ordered set.
\end{theorem}

\begin{corollary}
  \label{cor:min_decomp_dialects}
  Because each dialect decomposition corresponds to a monotonic decomposition according to Proposition \ref{prop:dialect_monotonic_decomp},
  Theorem \ref{thm:min_decomp} implies that a lower bound on the number of dialects given for a probability distribution as expressed by Equation \ref{eq:mixture_struct}
  is the minimum number of terms in a minimally refined monotonic decomposition.
\end{corollary}

Before attempting the proof of Theorem \ref{thm:min_decomp},
it helps to consider the case of monotonic functions before considering the general situation.

\begin{lemma}
  \label{lem:min_refined_special}
  Suppose that $f: P \to [0,\infty)$ is a monotonic function from a finite partially ordered set $(P,\le)$.
  If $(P,\le)$ has a unique minimal element $p$, then every monotonic decomposition of $f$ refines
  \begin{equation}
    \label{eq:min_refined_special}
    f = 1_{U_p} f.
  \end{equation}
  Additionally, this monotonic decomposition refines no other monotonic decomposition of $f$.
\end{lemma}
\begin{proof}
  For the first statement, suppose that
  \begin{equation*}
    f(x) = \sum_{k=1}^{N} 1_{U_{z_k}}(x) h_k(x).
  \end{equation*}
  Evidently, since $p$ is the minimal element of $(P,\le)$, it follows that $U_{z_k} \subseteq 1_{U_p} = P$.
  Moreover, since $h_k(x)$ is nonnegative and the collection sums to $f$, it follows that $h_k \le f$.

  Conversely, the only way that the monotonic decomposition defined by Equation \ref{eq:min_refined_special} refines any other is that $z_k =p$ for some $k$.
  If this is the case, refinement requires that $f \le h_k$.
  However, since the $h_k$ must sum to $f$, is still the case that $h_k \le f$.
  Hence $h_k = f$.
\end{proof}

\begin{proof} (of Theorem \ref{thm:min_decomp})
  Suppose that $f: P \to [0,\infty)$ is a function from a finite partially ordered set $(P,\le)$.
    
  Lemma \ref{lem:min_refined_special} can be used with Lemma \ref{lem:max_monotonic_lower_bound} to rule out refinement by certain decompositions supported on minimal elements.
  Suppose tentatively that $p \in P$ is a minimal element of $(P,\le)$.
  If $g_p: U_p \to [0,\infty)$ is the unique maximal monotonic decreasing function such that $g \le f$ guaranteed by Lemma \ref{lem:max_monotonic_lower_bound},
    then
    \begin{equation*}
      f = 1_{U_p} g_p + \sum_{k=1}^{N} 1_{U_{z_k}}(x) h_k(x)
    \end{equation*}
    cannot refine any monotonic decomposition which does not also contain the term $1_{U_p} g_p$.

With this situation treated, now consider the case where there are multiple minimal elements of $(P,\le)$.    
Let $p_1, \dotsc, p_m$ be all of the minimal elements of a finite partially ordered set $(P,\le)$.

Define monotonic decreasing functions $g_i : P \to [0,\infty)$ for each $i=1,\dotsc,m$ inductively via
  \begin{itemize}
  \item Base case: $g_1$ is the unique maximal monotonic decreasing function on $U_{p_1}$ such that $g \le f$ guaranteed by Lemma \ref{lem:max_monotonic_lower_bound},
  \item Induction case: $g_i$ is the unique maximal monotonic decreasing function on $U_{p_1}$ such that $g \le \left(f - \sum_{j=1}^i 1_{U_{p_i}} g_i\right)$ guaranteed by Lemma \ref{lem:max_monotonic_lower_bound}.
  \end{itemize}
  
  Any monotonic decomposition of the form
  \begin{equation*}
    f(x) = \sum_{j=1}^m 1_{U_{p_i}} g_i(x) + \sum_{k=1}^N 1_{U_{y_k}} h_k(x)
  \end{equation*}
  cannot refine any monotonic decomposition not containing all of the terms in the first sum.

  The reader is cautioned that the ordering of the $p_i$ in the above construction will generally yield different corresponding $g_i$ functions.
  The monotonic decompositions so arising cannot refine each other as a result.
  The tie-breaking procedure used in the proof of Proposition \ref{prop:decomp} provides one such option for an ordering.

The full statement of the Theorem follows by mimicking the induction case of the proof of Proposition \ref{prop:decomp}.
That is, we repeat the above procedure with $\left(f -  \sum_{j=1}^m 1_{U_{p_i}} g_i \right)$ instead of $f$,
and restrict the domain to $P \setminus \{p_1, \dotsc, p_m\}$.
At each iteration, we obtain more terms of the minimally refined monotonic decomposition.
\end{proof}

\subsection{The structure of the refinement preorder}
\label{sec:refinement}

In the previous section, it was shown (Lemma \ref{lem:decomp_order}) that monotonic decompositions are preordered by refinement.
We now establish that this preorder can be strengthened to a partial order (Proposition \ref{prop:decomp_partial_order}) if redundancies of a certain kind are eliminated.

\begin{lemma}
  \label{lem:irredundant}
  Suppose that there are two monotonic decompositions of a function $f: P \to [0,\infty)$ that refine each other,
  and that these two decompositions can be written as
  \begin{equation*}
    f(x) = \sum_{k=1}^{N} 1_{U_{z_k}}(x) h_k(x) = \sum_{j=1}^{M} 1_{U_{y_j}}(x) g_j(x).
  \end{equation*}
  If we assume that the sets $\{z_k\}$ and $\{y_j\}$ are antichains in $P$,
   then the two decompositions differ at most by a reordering of terms.
\end{lemma}
\begin{proof}
  Because the first monotonic decomposition refines the second,
  this means that every $z_k$ is greater than at least one element of $\{y_j\}$ in $(P,\le)$.
  In fact, this means there is (not uniquely) an order preserving function $r: \{z_k\} \to \{y_j\}$ such that $r(z_k) \le z_k$ for every $k$.
  On the other hand, the fact that the second monotonic decomposition refines the first means that every $y_j$ is greater than at least one element of $\{z_k\}$ in $(P,\le)$.
  Again, there exists an order preserving function $s: \{y_j\} \to \{z_k\}$ such that $s(y_j) \le y_j$ for every $j$.

  Using this notation,
  \begin{equation*}
    z_m = s(r(z_k)) \le r(z_k) = y_j \le z_k,
  \end{equation*}
  but being an antichain means that $z_m \le z_k$ implies $m=k$.
  Thus, $y_j = z_k$ as well.
  Hence the collection of support sets for the two decompositions coincide up to a permutation of indices.
  We can therefore compare the associated monotonic functions $h_k$ and $g_j$ on $U_{z_k}= U_{y_j}$.
  Since both decompositions refine each other, we have that $h_k \le g_j$ and $h_k \ge g_j$.
  Antisymmetry for $\le$ on functions completes the argument.
\end{proof}

The hypotheses for Lemma \ref{lem:irredundant} only yield a sufficient condition.
This condition is too restrictive,
because we often want to represent dialects that have subset behaviors of larger ones.
Such a situation is represented by a monotonic decomposition in which $z_m < z_k$ are both present.

\begin{lemma}
  \label{lem:zeroing}
  Suppose that
  \begin{equation*}
    f(x) = 1_{U_{y}}(x) g_1(x) + 1_{U_{z}}(x) g_2(x) = 1_{U_{y}}(x) h(x)
  \end{equation*}
  are two monotonic decompositions of $f: P \to [0,\infty)$ that refine each other,
    and that $y \le z$ in $P$.
    Then $g_2$ is identically zero.
\end{lemma}
\begin{proof}
  Because the left decomposition refines the right one, we have that $g_1 \le h$.
  On the other hand, because the right decomposition refines the left one, we also have that $h \le g_1$.
  Thus $g_1 = h$.
  Since $y \le z$ in $P$,
  this means that $U_z \subseteq U_y$.
  Therefore,
  \begin{equation*}
    \begin{aligned}
    f(z) &= 1_{U_{y}}(z) h(z) \\
    &= h(z) \\
    &= 1_{U_{y}}(z) g_1(z) + 1_{U_{z}}(z) g_2(z) \\
    &=g_1(z) + g_2(z) \\
    &=h(z) + g_2(z),
    \end{aligned}
  \end{equation*}
  whence $g_2(z) = 0$.  Since $g_2$ is assumed to be monotonic, this means that $g_2$ is identically zero.
\end{proof}

One irritation is that there can be redundancies that complicate minimality.

\begin{definition}
  A monotonic decomposition is called \emph{irredundant}
  \begin{equation*}
    f(x) = \sum_{k=1}^{N} 1_{U_{y_k}}(x) g_k(x)
  \end{equation*}
  if each of the $g_k$ functions is nonzero for at least one $x\in P$.  
\end{definition}

Simply by excluding any zero terms, every monotonic decomposition refines a unique irredundant monotonic decomposition.

\begin{proposition}
  \label{prop:decomp_partial_order}  
  The set of irredundant monotonic decompositions of a function $f: P \to [0,\infty)$ from a finite partially ordered set $(P,\le)$
    is a partially ordered set.
\end{proposition}
\begin{proof}
  All that remains after Lemma \ref{lem:decomp_order} is antisymmetry.
  Suppose that
  \begin{equation*}
    f(x) = \sum_{j=1}^{M} 1_{U_{y_j}}(x) g_j(x) = \sum_{k=1}^N 1_{U_{z_k}}(x) h_k(x)
  \end{equation*}
  are two irredundant monotonic decompositions that refine each other.
  We want to show that these two decompositions are in fact the same, up to reordering of terms.

  Let us establish that the sets $\{y_j\}$ and $\{z_k\}$ are identical.
  To see that the desired result follows from this statement,
  suppose that $y_j = z_j$ for some $j$.
  Then $g_j \le h_j$ and $h_j \le g_j$ by the refinement hypotheses, so $g_j = h_j$.

  Without loss of generality, suppose that there is a $y \in \{y_j\}$ that is not equal to any $z_k$.
  Since the first decomposition refines the second,
  this means that there must nevertheless be a $z_k$ such that $z_k \le y$.
  
  Discern two cases: either $z_k$ is equal to an element of $y_{j_1}$ or there is no such element.
  In the first case, Lemma \ref{lem:zeroing} asserts that $g_{j_1}=0$ in contradiction to the irredundancy of the first decomposition.
  
  In the second case, although $z_k$ is not equal to any element $y_m$,
  nevertheless refinement requires there to be a $y_{j_1}$ such that $y_{j_1} \le z_k \le y$.
  Assuming Lemma \ref{lem:zeroing} does not apply outright to this new situation,
  we can continue iterating this process to obtain a sequence $y \ge y_{j_1} \ge y_{j_2} \ge \dotsb$.
  Since $P$ is a finite set, this sequence must terminate at some $y' \in \{y_j\}$.
  Again, because both decompositions refine each other, we must conclude that $y' \in \{z_k\}$.
  Lemma \ref{lem:zeroing} applies to this situation, and thereby contradicts the irredundancy of at least one of the decompositions.
\end{proof}

\begin{definition}
  \label{def:redundant_cover}(standard, see for instance \cite{Kime_1974})
  Suppose that $(X,\col{T})$ is a topological space for which $\col{T}$ is finite,
  and that $\col{U} \subseteq \col{T}$ is a cover for $X$.

  An open set $U \in \col{U}$ is called \emph{redundant in $\col{U}$}
  if there is a subset $\col{V} \subset \col{U}$ such that $U \notin \col{V}$ but $U \subseteq \cup \col{V}$.
  
  A cover with no redundant open sets is called an \emph{irredundant cover}.
\end{definition}

Efficient algorithms for finding irredundant covers have been known for a long time \cite{Kime_1974}.
Irredundant monotonic decompositions correspond to dialects with nonzero probabilities, 
and so are useful in helping to identify candidate dialect decompositions.
It is informative to know the number of dialects that could be present in a given dataset,
which Corollary \ref{cor:min_decomp_dialects} relates to minimal monotonic decompositions.
From a practical matter, irredundant monotonic decompositions are especially useful because even though they are not unique,
they are unambiguous about the \emph{number} of dialects involved.

\begin{theorem}
  \label{thm:minimal_count}
  The number of dialects is bounded below by the number of support sets in the irredundant cover refined by any minimal irredundant monotonic decomposition of the joint probability distribution function.
\end{theorem}

Theorem \ref{thm:minimal_count} is an immediate consequence of two Lemmas, which follow.

\begin{lemma}
  \label{lem:minimal_count}
  All minimal irredundant monotonic decompositions of a given function have the same support sets and hence have the same number of terms.
\end{lemma}
This is different from the related situation of finding minimal irredundant decompositions of logic functions.
Logic functions are known that have minimal irredundant decompositions into sums with different numbers of terms \cite{minato1993fast}.
\begin{proof}
  Suppose that
  \begin{equation*}
    f(x) = \sum_{y\in R \subseteq P} 1_{U_y}(x) g_y(x)
  \end{equation*}
  is a minimal irredundant monotonic decomposition of an arbitrary function $f$.
  The statement to be proven is that the $R$ set in the equation is the same for all minimal irredundant monotonic decompositions of $f$.
  More explicitly, if $y \in R$,
  so that $1_{U_y}(x) g_y(x)$ is a term in the minimal irredundant monotonic decomposition above,
  then any other irredundant monotonic decomposition must also have a term of the form $1_{U_y}(x) h_y(x)$.

  Suppose that $y \in P$ is such that $f(x) = 0$ for all $x \le y$.
  Then every monotonic decomposition (irredundant or not) of $f$ must contain a term of the form $1_{U_y}(x) h_y(x)$.

  Suppose that $y \in P$ is such that there is an $x \in P$ such $x \le y$ and $f(x) \not= 0$.
  We can rewrite the monotonic decomposition as
  \begin{equation*}
    \begin{aligned}
      f(x) = &\sum_{v\in R : v < y} 1_{U_v}(x) g_v(x) + \sum_{w\in R : y \le w} 1_{U_w}(x) g_w(x) + \\
      &\sum_{z\in R : z \not\le y, y \not\le z} 1_{U_z}(x) g_z(x).
    \end{aligned}
  \end{equation*}
  Because $g_y(y) \not=0$ by irredundancy,
  the middle term must be positive.
  This means that the above decomposition leads to the inequality
  \begin{equation*}
    f(y) > \sum_{v\in R : v < y} 1_{U_v}(y) g_v(y).
  \end{equation*}
  
  We can take this inequality a bit further.
  The Proposition follows if there are is a subset $R'$ of those $v \in R$ satisfying both $v < y$ and
  \begin{equation}
    \label{eq:irredundant_minimal_count}
    f(y) > \sum_{v\in R : v < y} 1_{U_v}(y) g_v(y) = \sum_{u \in R'} f(u).
  \end{equation}
  This claim can be proven by contradiction;
  assume that there is a subset $R' \subseteq \{ v \in R: v < y\}$
  such that
  \begin{equation*}
    f(y) = \sum_{u \in R'} f(u).
  \end{equation*}
  If this is the case,
  using the fact that the $g_v$ functions sum to $f(u)$ on each $u \in R'$,
  we can choose the $g_v$ functions to take the same value at $y$ without violating monotonicity.
  Thus $g_y$ has to be zero because the sum of all the $g_v(y)$ is equal to $f(y)$, a contradiction with irredundancy.
  Obviously a smaller $f(y)$ forces $g_y(y)=0$ as well.
  In any case, this also contradicts irredundancy.
  Thus $f(y)$ is strictly greater than that, as Equation \eqref{eq:irredundant_minimal_count} claims.

  Notice that the last sum in Equation \eqref{eq:irredundant_minimal_count} does not depend on $g_v$.
  A term involving $U_y$ in an irredundant monotonic decomposition is therefore determined directly by the values of $f$.
  Therefore, any other monotonic decomposition will be subject to the same situation
  and therefore will need to contain a term involving $y$.
\end{proof}

\begin{lemma}
  \label{lem:irredundant_cover_bound}
  Suppose that $(X,\col{T})$ is a topological space for which $\col{T}$ is finite,
  that $\col{U}$ is a cover which refines an irredundant cover $\col{V}$.
  Then $\#\col{V} \le \#\col{U}$.
\end{lemma}
\begin{proof}
  Let $V \in \col{V}$.
  The hypotheses imply there is a $U \in\col{U}$ such that $U \subseteq V$.
  Let us establish this claim by contradiction.
  Suppose that no $U\in \col{U}$ is a subset of $V$.
  Because $\col{U}$ is a cover,
  there is a collection $\col{U}' \subseteq \col{U}$ such that
  $V \subseteq \cup \col{U}'$.
  Since $\col{U}$ refines $\col{V}$,
  each $U' \in \col{U}'$ is a subset of some $V' \in \col{V}$.
  Consider the subset
  \begin{equation*}
    \col{V}' := \{V' \in \col{V} : U' \subseteq V' \text{ for some } U' \in \col{U}' \} \subseteq \col{V}.
  \end{equation*}
  Evidently $V \subseteq \cup \col{V}'$.
  Recalling that we assumed no $U \in \col{U}$ is a subset of $V$,
  we must conclude that $V \not= \cup \col{V}'$,
  which contradicts the irredundancy of $\col{V}$.

  We complete the argument by induction on $\#\col{V}$.
  \begin{itemize}
  \item Base case: Suppose that $\#\col{V} = 1$.
    Because both $\col{V}$ and $\col{U}$ both cover $X$,
    and $\col{V}$ is evidently nonempty,
    then $\col{U}$ must also be nonempty.
    
  \item Induction case: Suppose that the Lemma has been established for all $\col{V}$ with $\#\col{V} \le n$ for some integer $n$.
    Suppose that $\col{V}$ is an irredundant cover containing $n+1$ elements, $V_0, \dotsc, V_n$.
    Consider the subspace of $(X,\col{T})$ covered by $V_0, \dotsc, V_{n-1}$.
    By the claim proven above,
    there is a subset $\col{U}' \subseteq \col{U}$
    that both covers $V_0 \cup \dotsb \cup V_{n-1}$
    and refines the cover $\{V_0, \dotsc, V_{n-1}\}$.
    The induction hypothesis applied to this situation asserts that $\# \col{U}' \ge n$.
    By the irredundancy of $\col{V}$,
    we must have that $V \not\subseteq (V_0 \cup \dotsb \cup V_{n-1})$.
    Because each element of $\col{U}'$ is a subset of at least one of the $V_0, \dotsc, V_{n-1}$,
    we have that $V \not\subseteq \cup \col{U}'$ as well.
    Therefore, to be a cover of $X$,
    $\col{U}$ must have at least one more element than $\col{U}'$.
    Hence,
    \begin{equation*}
      \# \col{U} \ge \#\col{U}' + 1 \ge n+1 = \# \col{V}. \qedhere
    \end{equation*}
  \end{itemize}
\end{proof}

\begin{proof} (of Theorem \ref{thm:minimal_count})
  This is an immediate consequence of Lemmas \ref{lem:minimal_count} and \ref{lem:irredundant_cover_bound}.
\end{proof}

\begin{example}
Suppose that $P = \{a,b,c,d\}$ is the partial order defined by the Hasse diagram
  \begin{equation*}
    \xymatrix{
      &d&\\
      b \ar[ur] & & c \ar[ul]\\
      &a \ar[ur]\ar[ul]&\\
      }
  \end{equation*}
  According to Theorem \ref{thm:minimal_count}, all minimally refined irredundant monotonic decompositions of an arbitrary function $f : P \to [0,\infty)$ have the same number of terms.
    We can demonstrate this fact by reasoning about monotonic decompositions directly.
    
  If $f(a) \not= 0$, then every irredundant monotonic decomposition of $f$ must contain a term of the form $1_{U_a} g_a$ where $g_a(a) = f(a)$.
  Therefore, without loss of generality, we may assume that $f(a) = 0$.
  Also, without loss of generality, we may assume that $f(b) \le f(c)$.

  If $f(b) = 0$, then there are no choices to be made in the decomposition of $f$:
  \begin{itemize}
  \item If $f(c) = 0$, the decomposition has at most one term of the form $1_{U_d} f(d)$.
  \item If $f(d) > f(c) > 0$ the decomposition contains a term of the form $1_{U_d} (f(d)-f(c))$.
  \item Otherwise the decomposition has only one term.
  \end{itemize}

  If instead $f(b) \not= 0$, whether we start the decomposition using $b$ or $c$ does not change the resulting number of terms in the decomposition.
  This happens because after removing the contribution from a term supported on $b$ or $c$ results in a new function that decomposes as above.
\end{example}

While Theorem \ref{thm:minimal_count} handles the case of minimally refined irredundant monotonic decompositions,
which have useful implications for determining the number of dialects,
\emph{maximally} refined irredundant monotonic decompositions also exist.

\begin{proposition}
  \label{prop:max_refined}
  There is a unique maximally refined irredundant monotonic decomposition of a function $f: P \to \mathbb{Z}^+$ from a finite partially ordered set $(P,\le)$ to the nonnegative integers, namely
    \begin{equation}
      \label{eq:max_refined}
      f(x) = \sum_{y \in P} \sum_{i=1}^{f(y)} 1_{U_y}(x) 1_{\{y\}}(x).
    \end{equation}
\end{proposition}
The maximally refined monotonic decomposition is rather uninformative,
because it means that each dialect contains exactly one file.
\begin{proof}
  Because $P$ is assumed to be finite and each of the monotonic functions in a monotonic decomposition produce nonnegative integers,
  the set of monotonic decompositions is finite.
Therefore, Lemma \ref{lem:decomp_order} implies that there are maximal and minimal monotonic decompositions under refinement.

  Suppose that we have an arbitrary monotonic decomposition of $f$,
  \begin{equation}
    \label{eq:arb}
    f(x) = \sum_{k=1}^{N} 1_{U_{y_k}}(x) g_k(x).
  \end{equation}
  We must show that Equation \eqref{eq:max_refined} refines this decomposition.

  Close inspection of the sum in Equation \eqref{eq:max_refined} reveals that it only contains terms involving $U_y$ if $f(y) > 0$.
  While the outer sum would seem to imply that terms involving $U_y$ will be present for all $y \in P$,
  the inner sum prevents the inclusion of any term for which $f(y) = 0$.
  Hence, Equation \eqref{eq:max_refined} defines an irredundant monotonic decomposition.

  Given this observation, consider a $y \in P$ for which $f(y) > 0$.
  Necessarily, there must be a term in Equation \eqref{eq:arb} for which $y \in U_{y_k}$ and $g_k(y) > 0$.
  Therefore, $U_y \subseteq U_{y_k}$.
  Moreover, since each of the monotonic functions in Equation \eqref{eq:max_refined} simply take the value $1$ on exactly one element of $P$,
  we have that $1_{\{y\}} \le g_k(y)$.  
\end{proof}

\begin{remark}
  If we instead permit redundancies in Proposition \ref{prop:max_refined}, then we may add terms with the zero function arbitrarily.  While these monotonic decompositions all refine each other, this precludes uniqueness of such a decomposition.
\end{remark}

\begin{remark}
  If we instead consider $f: P \to [0,\infty)$ in Proposition \ref{prop:max_refined}, then the inner sum in Equation \eqref{eq:max_refined} becomes infinite.  There is no maximally refined monotonic decomposition in this case.
\end{remark}

\end{document}